\documentclass{amsart}
\usepackage{graphicx}
\usepackage{subfigure}

\newtheorem{Teorema}{Theorem}
\newtheorem{Corolario}{Corollary}
\newtheorem{Lema}{Lemma}
\newtheorem{definition}{Definition}

\graphicspath{{img/}{.}{fig/}}
\usepackage[foot]{amsaddr}

\usepackage[linesnumbered,boxed,ruled]{algorithm2e}
\SetKwInput{KwInit}{Initialization}

\begin{document}

\title[Understanding edge-connectivity in the Internet]{Understanding
edge-connectivity in the Internet through core-decomposition}
\author{J. Ignacio Alvarez-Hamelin \and Mariano G. Beir\'o \and Jorge R.  Busch}

\address{Facultad de Ingenier\'{\i}a, Universidad de Buenos Aires, Paseo
Col\'on 850, C1063ACV Buenos Aires -- Argentina}
\email{ignacio.alvarez-hamelin@cnet.fi.uba.ar,\{mbeiro,jbusch\}@fi.uba.ar}

\begin{abstract}
Internet is a complex network composed by several networks: the
Autonomous Systems, each one designed to transport information
efficiently.  Routing protocols aim to find paths between nodes whenever
it is possible (i.e., the network is not partitioned), or to find paths
verifying specific constraints (e.g., a certain QoS is
required). As connectivity is a measure related to both of them
(partitions and selected paths) this work provides a formal lower bound
to it based on core-decomposition, under certain conditions, and low
complexity algorithms to find it. We apply them to analyze maps obtained
from the prominent Internet mapping projects, using the LaNet-vi
open-source software for its visualization. 
\end{abstract}

\keywords{connectivity, core-decomposition, algorithm, visualization}
\maketitle

\section{Introduction}

Nowadays, Internet is a highly developed network connecting people
around the world. Its continuous growth raises new problems, challenging
us to find novel and creative solutions. Modeling the Internet is a
problem as relevant as difficult: A tight model of the Internet would 
provide us with a powerful tool to analyze
applications behavior or test routing algorithms during their design. 
In the beginnings of the Internet \cite{Waxman88} was
the first to introduce an Internet model to test his multicast protocol.
Then other authors introduced models trying to improve Waxman's one adding hierarchy
and other sophistications (e.g., see
\cite{zegura97quantitative}). But the
\cite{faloutsos99powerlaw} paper was the first to show a remarkable
characteristic of the Internet's topology: the degree distribution of the nodes
follows a heavy-tailed behavior, instead of a Poisson's one as was believed (i.e. this was the case of the forementioned models which are a variant of the Random Graph \cite{ErdosRenyi59}). The unexpected consequence of this is that we have nodes with a large variety of degrees.

From this discovery, new simple models
aroused which tried to reproduce heavy-tailed behavior, such as the \cite{sf99} model based on preferential attachment, and later 
the \cite{FaKoPa02} model proving that a multivariate
optimization problem may lead to heavy-tailed degree distributions.

Later on, other properties were highlighted as important ones, like the
average neighboring degree distribution by \cite{PSVV01}. This
work showed that the Internet topology at the Autonomous System level (AS) is
different from the Internet topology at the inter-router level (IR).
On this basis, we think that a better understanding and modeling 
of the Internet's nature may be achieved by discovering new significant properties.

In this paper we turn our attention to connectivity in graphs, i.e., the number of different independent paths between their nodes. 
It is quite relevant as it represents the robustness
of the network: how tolerant it is to link failures, 
or what QoS it may offer when needed (e.g, see the RFCs of MultiProtocol Label
Switching \cite{MPLS1,MPLS2}). 
Also with regard to QoS, forwarding packets to more
connected nodes helps to find paths that fulfill the needed QoS,
therefore routing protocols like BGP \cite{BGP} might take advantage of information about connectivity.
Here we take different subgraphs (i.e., portions of the topology) and analyze each one's connectivity.

In particular, we introduce the concept of core-connected graphs. 
These graphs have the nice property that the minimum shell-index is a lower bound for the connectivity between two nodes.
As graphs are not core-connected in general, we also
present algorithms to find core-connected induced subgraphs. 
These algorithms are based on a generalization of Plesn\'{\i}k's
theorem \cite{plesnik:cgoagd}.

Finally, it has been shown previously that the $k$-core decomposition is capable of identifying networks sources by means of the
visualization \cite{AHDBV2006,BAHB2008}, may be used to validate models \cite{SBDG06}, and discover exploration biases on
the Internet \cite{AHDBV08}. These facts support the  $k$-core analysis as
one of the relevant tools to describe the Internet topology maps.

The paper is organized as follows. Section~\ref{sec:math} is devoted to
formalize the relation between $k$-cores and $k$-connected subgraphs.
Then in Section~\ref{sec:appl} we present a tool that implements these ideas giving a visualization
of the network according to its connectivity, and also a data file containing
the connectivity bound for every pair of nodes. We show some
applications to Internet maps.  Finally, we conclude with some important
remarks and comments on future work.

\section{Mathematics}\label{sec:math}

Let us introduce some general graph notions and notation.
Let $G$ be a simple graph ({\em i.e.}
a graph with no loops, no multiple edges) with vertex set $V(G)$ and edge set $E(G)$
(we follow in notation the book \cite{west:itgt}).
Given $A,B\subset V(G)$, $[A,B]$ is the set of edges
of the form $ab$,
joining a vertex $a\in A$ to a vertex $b\in B$.
As we consider edges without orientation,
$[A,B]=[B,A]$. Abusing of notation,
for $v\in V(G), A\subset V(G)$,
we write $[v,A]$ instead of $[\{v\},A]$.
The {\em degree} of a vertex $v\in V(G)$
is $d_G(v)\doteq |[v,V(G)]|$.
We shall denote 
$n(G)\doteq |V|,e(G)\doteq |E|,\delta(G)\doteq \min_{v\in V} d_G(v),
\Delta(G)\doteq \max_{v\in V} d_G(v)$.The {\em neighborhood} of a vertex $v$, $N(v)$,
is the set of vertexes $w$ such that $vw \in E(G)$.
Given $A\subset V(G)$, $G(A)$
is the graph $G'$ such that $V(G')=A$ and $E(G')$
is the set of edges in $E(G)$ having both endpoints in $A$.
Given $v,w\in V(G)$, $\rho_G(v,w)$
is the distance in $G$ from $v$ to $w$,
that is the minimum length of a path from $v$ to $w$.
If $v\in V(G),A\subset V(G)$ we set
$\rho_G(v,A)\doteq \min_{w\in A} \rho_G(v,w)$.
We shall also use 
the notation 
\[\rho_A\doteq \max_{a,b\in A}\rho_{G(A)}(a,b)\]
for the diameter of $G(A)$. 

We introduce some notions of connectivity used along the paper.
An edge cut in $G$ is a set of edges $[S,\bar{S}]$,
where $S\subset V(G)$ and $\bar{S}\doteq V(G)\setminus S$ 
are non void.

The edge-connectivity of $G$, $k'(G)$,
is the minimum cardinal of the cuts in $G$.
We say that $G$ is $k$-edge-connected if $k'(G)\geq k$.

Menger's theorem has as a consequence that, 
if $G$ is $k$-edge-connected, 
given two vertices 
$v,w$ in $V(G)$
there are at least $k$-edge-disjoint paths
joining $v$ to $w$ 
(see \cite{west:itgt}, pp.153-169).

As we shall not deal in this work with 
vertex-disjoint paths,
in the sequel we shall speak of $k$-connectivity,
avoiding the reference to the edges.

\subsection{Cores decomposition}
Let, for $A\subset V$, 
\[\psi(A)=\min_{v\in A} d_{G(A)} (v)\]
Notice that for $A_1,A_2\subset V$,
\[
\psi(A_1\cup A_2)\geq \min(\psi(A_1),\psi(A_2))
\]
and as a consequence for any $k$
\[
C_k=\cup \{A:\psi(A)\geq k\}
\]
satisfies $\psi(C_k)\geq k$,
and of course $C_k$ contains any other set that 
satisfies this property.

The preceding remark justifies 
the following \cite{Batagelj02}
\begin{definition} 
\label{k-core} 
A subgraph $H = G(C_k)$ induced by the set $C_k\subseteq V$ is the {\em
$k$-core} (or the core of order $k$) in $G$ if 
$C_k$ is the maximal subset of $V$ such that 
$\min_{v\in C} d_{G(C)}(v)\ge k$.
\end{definition}

Thus, if we let 
$k_{max} \doteq \max \{k: C_k\not=\emptyset\}$,
we obtain the decomposition
\[
V=\cup \{C_k :0 \leq k \leq k_{\max}\}
\]  
where $C_{k+1}\subset C_k, 0\leq k \leq k_{\max}-1$.

A $k$-core of $G$ can be obtained by recursively removing all
the vertices of degree lower than $k$, with their incident edges,
until all the vertices in the
remaining graph have degree greater than or equal to $k$. This
decomposition can be easily implemented: 
the algorithm by \cite{Batagelj03} has a time complexity of order 
$O(n(G)+e(G))$
for a general simple graph $G$. 
This makes the algorithm very efficient for sparse
graphs, where $e(G)$ is of the same order that $n(G)$.

\begin{definition}
\label{shell-index}
Let 
$S_k\doteq C_k\setminus C_{k+1}, 0\leq k \leq k_{\max}-1$,
$S_{k_{\max}}=C_{k_{\max}}$. 
We call $S_k$ the $k$-shell of $G$,
and if $v\in S_k$ we say that $v$ has {\em shell-index} $k$,
$sh(v)=k$.
\end{definition} 

Notice that $C_k$ is thus the union of all shells $S_s$ 
with $s\ge k$, and that the shells are pairwise disjoint.

\begin{definition} 
Every connected component of $S_s$ will be called {\em cluster}.
\end{definition} 

Each shell $S_s$ is thus composed by clusters $Q_s^m$, 
such that 
$S_s=\cup \{Q_s^m :1\leq m \leq q_{\max}(s)\}$, 
where $q_{\max}(s)$ is 
the number of clusters in $S_s$.

In this paper we address the following expansion problem: 
given a $k$-edge-connected graph $G_2$,
give conditions under which the result of adjoining
to $G_2$ a graph $G_1$ will be also $k$ edge-connected
(see Corollary \ref{conn3} below).  

\subsection{$k$-connectivity.} 
We may analyze connectivity in a strict or wide sense.

\begin{definition}
Let $A\subset V(G)$.
\begin{enumerate}
\item
We say that $A$ is $k$-connected in strict sense
if $G(A)$ is $k$-connected, {\em i.e} every cut
in $G(A)$ has at least $k$ edges. That is,
given  $u,v\in A$, there exist at least $k$
edge disjoint paths form $u$ to $v$ in $G(A)$.
\item
We say that $A$ is $k$-connected in wide sense if 
every cut $[X,\bar X]$
in $G$ such that $X\cap A\not =\emptyset$ and 
$\bar X\cap A\not =\emptyset$
has at least $k$ edges. That is,
given  $u,v\in A$, there exist at least $k$
edge disjoint paths form $u$ to $v$ in $G$.
\end{enumerate}
\end{definition}
Of course, if $A$ is $k$-connected
in strict sense, it is also $k$-connected in wide sense.

\begin{Lema}
Let $A,B\subset V$ and $A\cap B\not = \emptyset$. Then
\begin{enumerate}
\item
If $A$ and $B$ are $k$-connected in strict sense, so is $A\cup B$.
\item
If $A$ and $B$ are $k$-connected in wide sense, so is $A\cup B$.
\end{enumerate}
\end{Lema}
\begin{proof}\hfill
\begin{enumerate}
\item
Let $u,v \in A\cup B$. 
If $u,v\in A$ or $u,v\in B$, then there are $k$
disjoint paths in $G(A)$ or in $G(B)$ from $u$ to $v$,
and in any case there are $k$ disjoint paths in $G(A\cup B)$. 
Suppose then that  $u\in A$, $v\in B$,
and let $[X,\bar X]$ be a cut in $G(A\cup B)$ with 
$u\in X, v\in \bar X$. Suppose that $s\in A\cap B$ and,
without loss of generality, assume that $s\in X$.
Then there are $k$ disjoint paths in $G(B)$ from $s$ to $v$,
whence $|[X,\bar X]|\geq k$.
\item
Let $u,v \in A\cup B$. If $u,v\in A$ or $u,v\in B$, then there are $k$
disjoint paths in $G$ from $u$ to $v$. 
Suppose then that  $u\in A$, $v\in B$,
and let $[X,\bar X]$ be a cut with 
$u\in X, v\in \bar X$. Suppose that $s\in A\cap B$ and,
without loss of generality, assume that $s\in X$.
Then there are $k$ disjoint paths in $G$ from $s$ to $v$,
whence $|[X,\bar X]|\geq k$.
\end{enumerate}
\end{proof}
This lemma has as a consequence that given $v\in V$,
\[
\cup\{A :A\text{ is $k$-connected in strict sense and } v\in A\} 
\]
is $k$ connected in strict sense (and we call it the $k$-connected
component of $v$ in strict sense), and that
\[
\cup\{A :A\text{ is $k$-connected in wide sense and } v\in A\} 
\]
is $k$ connected in wide sense (and we call it the $k$-connected
component of $v$ in wide sense).
\begin{Lema}
For $k=1,2$ the $k$-connected components in strict sense
and in wide sense are the same. This is not true
for $k\geq 3$.
\end{Lema}

\subsection{An expansion theorem }
Let $G$ be a simple graph.
Let $Q,C\subset V(G)$, 
and set $C'\doteq Q\cup C$, 
$G'\doteq G(C')$. 
We assume in the sequel that $Q$ and $C$ are non void and 
that $Q\cap C=\emptyset$.
We define, for $x,y\in Q$, the {\em contracted distance}
\[
\rho_{C'/C}(x,y)\doteq 
\min\{
\rho_{G(Q)}(x,y),
\rho_{G'}(x,C)+\rho_{G'}(y,C)
\}\\
\]
and for $x\in C',y\in C$
\[
\rho_{C'/C}(x,y)=\rho_{C'/C}(y,x)\doteq \rho_{G'}(x,C)
\]
If $x\in C'$ and $A\subset C'$,
we set $\rho_{C'/C}(x,A)\doteq \min_{a\in A} \rho_{C'/C}(x,a)$.
We shall use also the notation
\[
\rho_{C'/C}\doteq\max_{x,y\in C'}\rho_{C'/C}(x,y)
\]
for the {\em contracted diameter} of $C'$.

Notice that with these definitions,
if $\rho_{C'/C}(x,y)=2$ for some $x,y\in C'$, 
then there exists $z\in C'$ 
such that $\rho_{C'/C}(x,z)=\rho_{C'/C}(z,y)=1$.
Notice also that $\rho_{C'/C}$ is,
whence the notation,
the pseudo-distance induced in $C'$
by the distance in the quotient graph
$G'/G(C)$, obtained from $G'$ by contracting
$C$ to a vertex (this quotient graph is not
a simple graph). 

We shall also use the notations 
\begin{eqnarray*}
\partial^j Q&\doteq&\{x\in Q: |[x,C]|\geq j\}\\
\bar{\partial}^j Q&\doteq& \{x\in Q: |[x,C]|< j\}=
Q\setminus\partial^j Q
\end{eqnarray*}

Under these settings, we consider also
\[
\Phi_{C'/C}\doteq\sum_{x\in Q}
\min\{\max\{1,|[x,\bar{\partial}^2 Q]|\},|[x,C]|\}
\]
As $Q$ and $C$
will be fixed in the following of this Section,
we shall freely omit the subindex ${C'/C}$
when speaking of $\rho_{C'/C}$ and $\Phi_{C'/C}$. 
For $v\in C'$, $N'(v)$ denotes its neighborhood in $G'$.

In this general framework, we have
\begin{Teorema}
\label{cuts1}
If $\rho_{C'/C}\leq 2$,
$[S,\bar{S}]$ is an edge cut in $G'$ such that
$C\subset S$, and we let $S_1\doteq S\cap Q$
\begin{enumerate}
\item
\label{110}
If $\max _{\bar{s}\in \bar{S}} \rho(\bar{s},S)=1$, 
then $|[S,\bar{S}]|\geq \max _ {\bar{s}\in \bar{S}} |N'(\bar{s})|$.
\item
\label{120}
If $\max _{\bar{s}\in \bar{S}} \rho(\bar{s},S)=1$, 
then $|[S,\bar{S}]\geq |\bar{S}|$.
\item
\label{130}
If $\max _{\bar{s}\in \bar{S}} \rho(\bar{s},S)=2$, 
then $|\bar{S}|>\min_{\bar{s}\in\bar{S}} |N'(\bar{s})|$.
\item
\label{140}
If $\max _{\bar{s}\in \bar{S}} \rho(\bar{s},S)=2$, 
then $\max _{{s}\in {S}} \rho({s},\bar{S})=1$.
\item
\label{150}
If $\max _{{s}\in {S_1}} \rho({s},\bar{S})=1$,
then $|[S_1,\bar{S}]|\geq max _{s\in S_1} (|N'(s)|-|N'(s)\cap C|)$.
\item
\label{160}
If $\max _{{s}\in {S_1}} \rho({s},\bar{S})=1$,
then $|[S_1,\bar{S}]|\geq |S_1|$.
\end{enumerate}
\end{Teorema}
\begin{proof}\hfill
\begin{enumerate}
\item
Suppose that 
for any $\bar{s}\in \bar{S}$: $\rho(\bar{s},S)=1$.
Let $\bar{s}\in\bar{S}$.
Then we have $k_1$ edges $\bar{s}s_i,1\leq i\leq k_1$
with $s_i\in S$ and (eventually) $k_2$ edges $\bar{s}\bar{s}_j$,  
$\bar{s}_j\in \bar{S}$.
But each $\bar{s}_j$ satisfies $\rho(\bar{s}_j,S)=1$,
thus we have $k_2$ new edges (here we used that $G'$ is simple,
because we assumed that the vertices $\bar{s}_j$ are different)
$\bar{s}_js'_j$, with $s'_j\in S$,
whence 
\[
|[S,\bar{S}]|\geq k_1+k_2=|N'(\bar{s})|
\]
\item
This follows at once if we notice
that in this case for each $\bar{s}\in \bar{S}$
there is at least one $s\in S$ such that $\bar{s}s\in [\bar{S},S]$.
\item
Notice that if $\rho(\bar{s},S)=2$,
then $\{\bar{s}\}\cup N'(\bar{s})\subset \bar{S}$.
\item
Let $\bar{s}_0\in \bar{S}$ be such that $\rho(\bar{s}_0,S)=2$.
Then for each $s\in S$, as  $\rho(\bar{s}_0,s)=2$,
there exists $\bar{s}'$ such that 
$\rho(\bar{s}_0,\bar{s}')=\rho(\bar{s}',s)=1$.
But, again, as $\rho(\bar{s}_0,S)=2$, 
it follows that $\bar{s}'\in \bar{S}$,
hence $\rho(s,\bar{S})=1$.
\item
For each $s\in S_1$ we have 
\[
N'(s)=(N'(s)\cap \bar{S})\cup (N'(s)\cap S_1) \cup (N'(s)\cap C)
\]
and, if $\max _{s\in S_1}\rho(s,\bar{S})=1$,
then for each $s'\in N'(s)\cap S_1$ we have at least 
one edge in $[s',\bar{S}]$,
thus $|N'(s)|\leq |[S_1,\bar{S}]|+|(N'(s)\cap C)|$.
\item
Our last statement follows noticing that 
$\max _{s\in S_1}\rho(s,\bar{S})=1$ means that
for $s\in S_1$ $|[s,\bar{S}]|\geq 1$,
and these sets are pairwise disjoint 
subsets of $[S_1,\bar{S}]$.
\end{enumerate}
\end{proof}

\begin{figure}[h!]
  \subfigure[]{
    \includegraphics[width=0.3\textwidth]{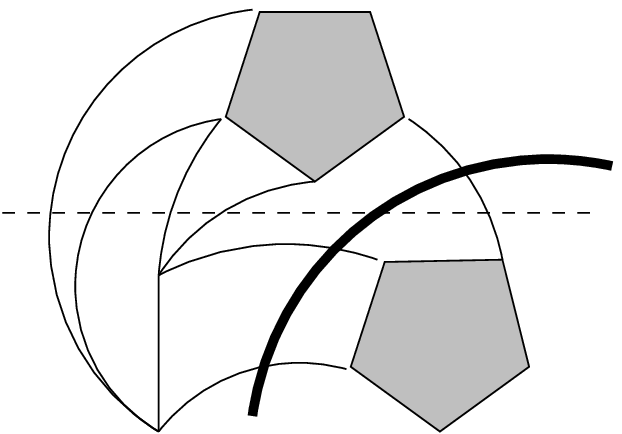}
  }
  \subfigure[]{
    \includegraphics[width=0.3\textwidth]{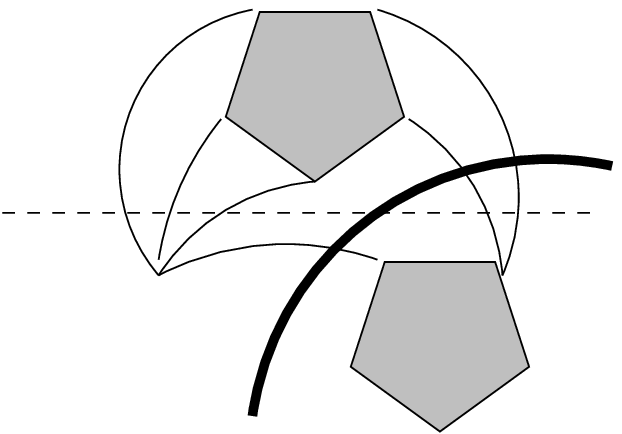}
  }\\
  \subfigure[]{
    \includegraphics[width=0.3\textwidth]{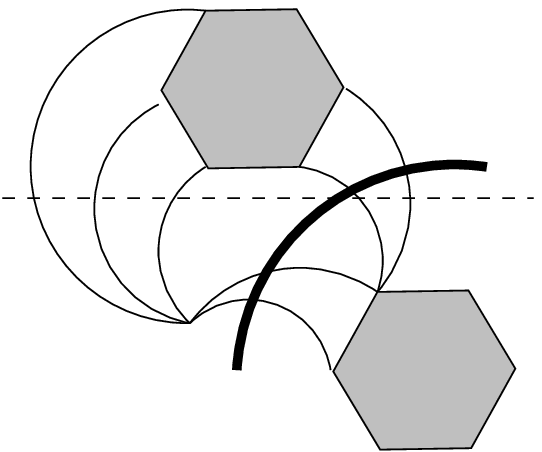}
  }
\hspace{0.1\textwidth}
  \subfigure[]{
    \includegraphics[width=0.22\textwidth]{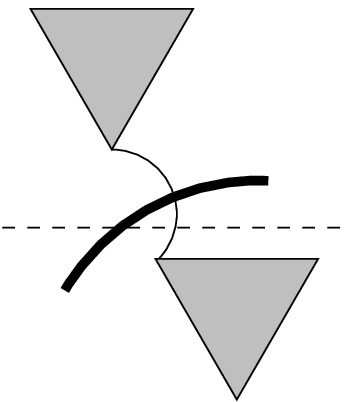}
  }
  \caption{\label{thm.ex} 
    {\bf Conventions:}
    1.Filled polygons represent cliques, and curved arcs represent edges.
    2.The dotted line separates 
    $C$ (the upper set of vertices) from 
    $Q$ .
    3. The widest arc shows the cut $[S,\bar{S}]$.
    4. $k=\min_{v\in Q} |N'(v)|$.
    {\bf Descriptions:} 
(a) Here $|[S,\bar{S}]|=3<k=4$, $|S\cap Q|=2$,
$|\bar{S}|=5$, $\Phi=3$, $S\cap Q=\partial^2 Q$.
(b) Here $|[S,\bar{S}]|=3<k=4$, $|S\cap Q|=1$,
$|\bar{S}|=5$, $\Phi=3$, $S\cap Q\not =\partial^2 Q$. 
(c) Here $|[S,\bar{S}]|=4<k=5$, $|S\cap Q|=1$,
$|\bar{S}|=6$, $\Phi=3$, $S\cap Q\not =\partial^2 Q$. 
(d) Here $|[S,\bar{S}]|=1<k=2$, $|S\cap Q|=0$,
$|\bar{S}|=3$, $\Phi=1$, $S\cap Q =\partial^2 Q=\emptyset$. 
}
\end{figure}

\begin{Corolario}
\label{cuts2}
Assume that in addition to the hypotheses of Theorem \ref{cuts1},
we have $|[S,\bar{S}]|<\min _{v\in Q} |N'(v)|$. Then
\begin{enumerate}
\item
\label{210}
$\max_{\bar{s}\in \bar{S}}\rho(\bar{s},S)=2$.
\item
\label{220}
$\max_{ s\in S} \rho(s,\bar{S})=1$.
\item
\label{230}
$|[\bar{S},C]|\geq 1$.
\item
\label{240}
$|S_1|<|[S,\bar{S}]|<\min_{v\in Q} |N'(v)|<|\bar{S}|$.
\item
\label{250}
$S\cap Q\subset \partial^2 Q, \bar{S}\supset \bar{\partial}^2 Q$.
\item
\label{260}
$
\Phi \leq |[S,\bar{S}]|
$.
\end{enumerate}
\end{Corolario}
(See the examples in Figure \ref{thm.ex}.)
\begin{proof}
Points \ref{210} and \ref{220} are obvious consequences 
of our new hypothesis and points \ref{110} and \ref{140}
in Theorem \ref{cuts1}.
To show point \ref{230}, 
notice that $\max_{ s\in S} \rho(s,\bar{S})=1$.
The first (from left to right) of the inequalities 
stated in point \ref{240}
follows
from point \ref{160} in Theorem \ref{cuts1},
and point \ref{230} in the present Theorem. 
The second of these inequalities is assumed 
by our additional hypothesis,
and the third follows immediately 
from point \ref{130} in Theorem \ref{cuts1}.

From point \ref{150} in Theorem \ref{cuts1} 
and $|[C,\bar{S}]|\geq 1$
we obtain \[|[S,\bar{S}]|>
max _{s\in S_1} (|N'(s)|-|N'(s)\cap C|)\]
Thus for $s\in S_1$
\[
N'(s)>|[S,\bar{S}]|>(|N'(s)|-|N'(s)\cap C|)
\]
whence $|N'(s)\cap C|\geq 2$. 
Point \ref{250} follows immediately from this.

By our previous points,
if $s\in S\cap Q$ then
\[
|[s,\bar{S}]|\geq \max\{1,|[s,\bar{\partial}^2Q]|\}
\]
and of course for $\bar{s}\in \bar{S}$, 
$|[\bar{s},S]|\geq |[\bar{s},C]|$,
thus
\begin{eqnarray*}
|[S,\bar{S}]|&=&|[S\cap Q,\bar{S}]|+|[\bar{S},C]|\\
&\geq&\sum_{s\in S\cap Q}\max\{1,|[s,\bar{\partial}^2Q]|\}
+\sum_{\bar{s}\in \bar{S}} |[\bar{s},C]|\\
&\geq&
\Phi
\end{eqnarray*}
\end{proof}

\begin{figure}[h!]
  \subfigure[]{
    \includegraphics[width=0.3\textwidth]{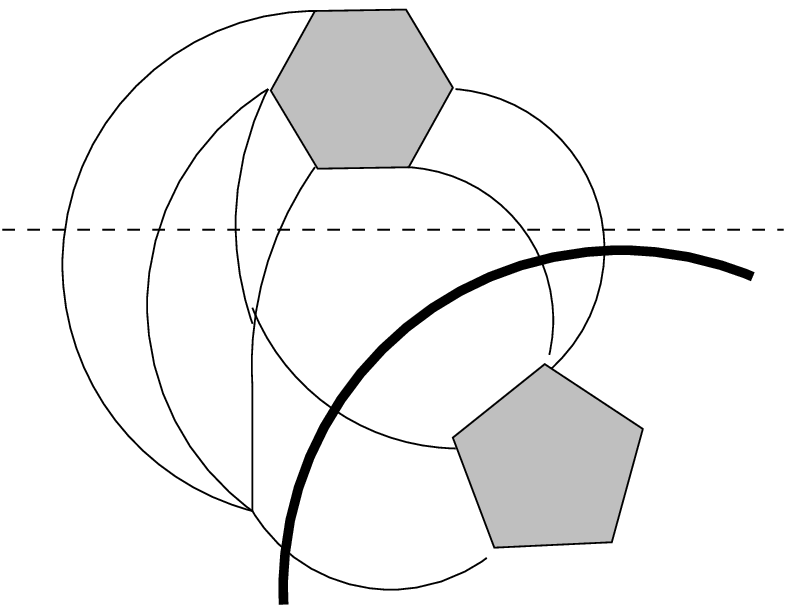}
  }
  \subfigure[]{
    \includegraphics[width=0.3\textwidth]{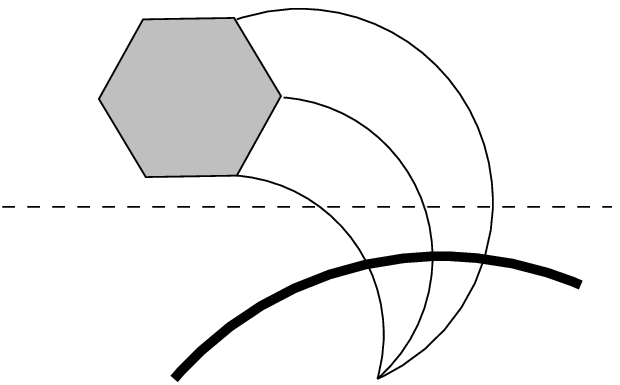}
  }
  \caption{\label{crll.ex} 
    {\bf Conventions:}
    1.Filled polygons represent cliques, and curved arcs represent edges.
    2.The dotted line separates 
    $C$ (the upper set of vertices) from 
    $Q$.
    3. The widest arc shows a minimal cut $[S,\bar{S}]$.
    {\bf Descriptions:} 
(a) Here $|[S,\bar{S}]|=k=4$, $\Phi=4$, $|\partial^1 Q|=3$.
(b) Here $|[S,\bar{S}]|=k=3$, $\Phi=1$, $|\partial^1 Q|=1$,
$Q=\partial^1 Q$. 
This example shows that Corollary \ref{conn3} includes an 
edge-connectivity version
of the Expansion Lemma in \cite{west:itgt}, Lemma 4.2.3.
}
\end{figure}

\begin{Corolario}
\label{conn3}
Let $k\leq\delta(G')$, and assume that 
\begin{enumerate}
\item
\label{320}
$G(C)$ is $\delta(G')$-edge connected
\item
\label{340}
$\rho_{C'/C}\leq 2$
\end{enumerate}
Then any of the following 
\begin{enumerate}
\item
$\Phi\geq k$
\item
$|\partial^1 Q|\geq k$ 
\item
$Q=\partial^1 Q$
\end{enumerate}
implies that  $G'$ is $k$-edge-connected.
\end{Corolario}
(See the examples in Figure \ref{crll.ex}.)

\begin{proof}
Let $[S,\bar{S}]$ be any cut in $G'$. 
We shall show that,
under the listed hypotheses and any of
the alternatives, $|[S,\bar{S}]|\geq k$.

If $S\cap C\not =\emptyset$ and 
$\bar{S}\cap C\not = \emptyset$,
then, as
\[
[S\cap C,\bar{S}\cap C]\subset[S,\bar{S}]
\]
is a cut in $G(C)$, which we assumed to be $k$-edge connected,
we obtain $|[S,\bar{S}]|\geq k$.

Without loss of generality, we assume in the sequel that
$C\subset S$. We argue by contradiction
assuming that there exists some $S$ such that $|[S,\bar{S}]|<k$,
so that we are under the hypothesis of Corollary \ref{cuts2}.

The first of our alternative hypothesis contradicts
point \ref{260} in the conclusions of Corollary \ref{cuts2}.

When $v\in \partial^1Q$, 
\[
\min\{\max\{1,|[v,\bar{\partial}^2 Q]|\},|[v,C]|\}\geq 1\] 
so that we have
$
|\partial^1 Q| \leq\Phi
$,
{\em i.e.} the second of our alternative hypothesis 
implies the first one.

To finish our proof,
notice that if $Q=\partial^1 Q$, as $\bar{S}\subset Q$,
we have $\rho(\bar{s},S)=1$ for any $\bar{s}\in \bar{S}$,
contradicting point \ref{210} in the conclusions 
of Corollary \ref{cuts2}.
\end{proof}
\begin{definition}
As this Corollary will be a key for our later results,
we shall set for future reference
\[
\Psi_{C'/C}(k,G)\doteq \max\{
\Phi_{C'/C}-k,
|\partial^1 Q|-k,
|\partial^1 Q|-|Q|
\},\quad \text{for }k\leq\delta(G')
\]
so that the validity of some of the three last alternative hypotheses
in Corollary~\ref{conn3} could be re-stated as 
$\Psi_{C'/C}(k,G)\geq 0$. 
\end{definition}
\subsubsection{Remarks}
Corollary \ref{conn3} is 
related to a well known theorem of Plesn\'{\i}k 
(see \cite{plesnik:cgoagd}, Theorem 6),
which states that in a simple graph
of diameter $2$ the edge connectivity
is equal to the minimum degree.

\subsection{Edge-connectivity and cores decomposition}

Let us introduce the following
\begin{definition} 
\label{core_x}
Consider a graph $G$, its cores decomposition 
(see Definition \ref{k-core})
\[
V(G)=\cup \{C_k: 0\leq k \leq k_{\max}\}
\]
and $A\subset V(G)$.
\begin{enumerate}
\item
We say that $A$ is $k$-core connected in strict sense
if $A\cap C_k$ is $k$-connected in strict sense.
\item
We say that $A$ is $k$-core connected in wide sense
if $A\cap C_k$ is $k$-connected in wide sense in $G(C_k)$.
\item
We say that $A$ is core connected in strict sense
if $A\cap C_k$ is $k$-connected in strict sense for all $k$
such that $A\cap C_k\not =\emptyset$.
\item
We say that $A$ is core connected in wide sense
if $A\cap C_k$ is $k$-connected in wide sense in $G(C_k)$ for all $k$
such that $A\cap C_k\not =\emptyset$.
\item
\label{core_x:k}
We say that $G$ is $k$-core connected
when $V(G)$ is $k$-core connected.
\item
\label{core_x:graph}
We say that $G$ is core-connected
when $V(G)$ is core-connected.
\end{enumerate}
\end{definition}

Next, we shall describe two algorithms that 
provide a mechanism to build (hopefully big) 
core-connected sets of vertices,
in strict and wide sense respectively.
Both algorithms proceed recursively,
starting from the highest core,
looking for a not-yet-joined cluster 
able to be joined,
joining it, and restarting from this new set of vertices.
The difference between the algorithms
lies in the meaning of ``able to be joined''.    

When $Q$ is a cluster, we denote $k(Q)$ its shell index. 
Let $\mathfrak{Q}$ be a family of clusters,
and let $k(\mathfrak{Q})\doteq\max\{k(Q):Q\in\mathfrak{Q}\}$ 
denote
the maximum shell index of the clusters in $\mathfrak{Q}$.

\begin{algorithm}[H]
  \caption{strict sense core-connected}\label{ss_algorithm}
  \SetLine
  \KwIn{$\mathfrak{Q}$, the family of all clusters of the k-core decomposition of a graph $G$}
\KwOut{$C\subset V$, core-connected in strict sense}
  \KwInit{$C\leftarrow \emptyset$, $k\leftarrow k_{max}$}
  \Begin{
     \While {$C = \emptyset$ {\bf and} $\mathfrak{Q} \neq \emptyset$ {\bf and} $k \geq 2$}{
        $k\leftarrow k(\mathfrak{Q})$ \; \label{rest0} 
        \If {there is some $Q\in\mathfrak{Q}$ satisfying: 
	  $\left\{
          \begin{array}{l}
          k(Q)=k \\
          \rho_Q\leq 2 \\
          \end{array} \right]$
          }{
          $C\leftarrow C\cup Q$ \;
        }
        $\mathfrak{Q}\leftarrow \mathfrak{Q}\setminus \{Q\in\mathfrak{Q}:k(Q)=k\}$ \; \label{qfirst}
     } \label{endloop1}
     \While{$\mathfrak{Q} \neq \emptyset$ {\bf and} $k \geq 2$}{  
          $k\leftarrow k(\mathfrak{Q})$ \; \label{rest1} 
          \While {there is some $Q\in\mathfrak{Q}$ satisfying:
             $\left\{ 
              \begin{array}{l}
	      k(Q)=k \\
	      \rho_{C\cup Q/C}\leq 2 \\
	      \Psi_{C\cup Q/C}(k,G)\geq 0 \\
	      \end{array}
	     \right]$ \label{qselec}  \label{loop3}
	     }{
             $C\leftarrow C\cup Q$ \;
	     $\mathfrak{Q}\leftarrow \mathfrak{Q}\setminus \{Q\}$ \;
          }
          $\mathfrak{Q}\leftarrow \mathfrak{Q}\setminus \{Q\in\mathfrak{Q}:k(Q)=k\}$ \; \label{elim}
     } \label{fin}
  }
\end{algorithm}
(see Figure~\ref{alg.il} for illustration)

\begin{Teorema}
Algorithm~\ref{ss_algorithm} always stops,
and when it stops the set $C$ is core connected in strict sense.
\begin{proof} \hfill
\begin{itemize}
\item For the first {\bf while} loop, step~\ref{rest0} computes the maximum
$k$ for the actual family of clusters, 
while step \ref{qfirst} deletes from $\mathfrak Q$
all clusters with shell index $k$.  
As a consequence, $k$ is strictly
decreasing, and when the algorithm arrives to step~\ref{endloop1}, either
$\mathfrak{Q}$ is empty or $C$ has a cluster verifying the hypothesis of
Corollary~\ref{conn3}.
\item The second {\bf while} loop will also finish because 
steps~\ref{rest1} and~\ref{elim} assure that $k$ is strictly decreasing.
\item The nested {\bf while} loop will finish 
because the family  $\mathfrak{Q}$ is finite. 

\item 
Assume that the actual $C$ is core connected
in strict sense when we arrive to step~\ref{qselec}.
By construction, $C\cap C_j=C$ when $j\leq k$,
and $C\cap C_j$ are previous instances of $C$ when $j>k$
(in fact, this instances are obtained each time
that we arrive to step~\ref{fin}).
The new $C$, let us call it $C'$ for a moment, 
has the same intersections with $C_j$
for $j>k$, 
and when $j\leq k$ the intersection is $C'=C\cup Q$,
that is $k$ connected in strict sense by
Corollary~\ref{conn3},
as the conditions for the selection of $Q$ in step~\ref{qselec}
match the hypothesis of Corollary~\ref{conn3}. 
Thus, all the instances of $C$
during the algorithm are core connected in
strict sense, 
whence the final $C$ is core connected in strict sense.
\end{itemize}
\end{proof}
\end{Teorema}

\begin{algorithm}[H]
  \caption{wide sense core-connected}\label{ws_algorithm}
  \SetLine
  \KwIn{$\mathfrak{Q}$, the family of all clusters of a graph $G$}
  \KwOut{$C\subset V$, core-connected in wide sense}
  \KwInit{$C\leftarrow\emptyset$ , $D\leftarrow \emptyset$ , $\mathfrak{Q'\leftarrow \emptyset}$, $k\leftarrow k_{max}$}
  \Begin{
     \While {$C = \emptyset$ {\bf and} $\mathfrak{Q} \neq \emptyset$  {\bf and} $k \geq 2$}{
       $k\leftarrow k(\mathfrak{Q})$ \;\label{kreset}
       \If {there is some $Q\in\mathfrak{Q}$ satisfying: 
	 $\left\{
         \begin{array}{l}
         k(Q)=k \\
         \rho_{Q}\leq 2 \\
         \end{array}
	 \right]$
	 }{
         $C\leftarrow C\cup Q$ \;
         $\mathfrak{Q}\leftarrow \mathfrak{Q}\setminus \{Q\}$ \;
         }
         $\mathfrak{Q'}\leftarrow \mathfrak{Q'}\cup\{Q\in\mathfrak{Q}:k(Q)=k\}$ \;
         $\mathfrak{Q}\leftarrow \mathfrak{Q}\setminus \{Q\in\mathfrak{Q}:k(Q)=k\}$ \;\label{Qdeletion}
       }\label{ws.endloop1}
     \While{$\mathfrak{Q} \neq \emptyset$ {\bf and} $k \geq 2$}{
         $k\leftarrow k(\mathfrak{Q})$ \;\label{kreset2}
         \While {there is some $Q'\in\mathfrak{Q'}$ satisfying: 
	     $\left\{
             \begin{array}{l}
             k(Q')\geq k \\
             \rho_{(C\cup D\cup Q')/(C\cup D)}\leq 2 \\
             \Psi_{(C\cup D\cup Q')/(C\cup D)}(k,G)\geq 0 \\
	     \end{array}
	     \right]$
	     }{
             $D\leftarrow D\cup Q'$ \;
	     $\mathfrak{Q'}\leftarrow \mathfrak{Q'}\setminus \{Q'\}$ \;
         }
         \While {there is some $Q\in\mathfrak{Q}$ satisfying: 
	     $\left\{
             \begin{array}{l}
             k(Q)=k \\
             \rho_{(C\cup D\cup Q)/(C\cup D)}\leq 2 \\
             \Psi_{(C\cup D\cup Q)/(C\cup D)}(k,G)\geq 0 \\
	     \end{array}
	     \right]$\label{ws.qselec}
	     }{
             $C\leftarrow C\cup Q$ \;
	     $\mathfrak{Q}\leftarrow \mathfrak{Q}\setminus \{Q\}$ \;
         }
         $\mathfrak{Q'}\leftarrow \mathfrak{Q'}\cup\{Q\in\mathfrak{Q}:k(Q)=k\}$ \;
         $\mathfrak{Q}\leftarrow \mathfrak{Q}\setminus \{Q\in\mathfrak{Q}:k(Q)=k\}$\;\label{Qdeletion2}
     }\label{ws.endloop2}
  }
\end{algorithm}
(see Figure~\ref{alg.il} for illustration)
\begin{Teorema}
Algorithm~\ref{ws_algorithm} always stops,
and when it stops the set $C$ is core connected in wide sense.
\end{Teorema}
\begin{proof} \hfill
\begin{itemize}
\item For the first {\bf while} loop, step~\ref{kreset} 
computes the maximum
$k$ for the actual family of clusters, while step \ref{Qdeletion} 
deletes from the actual $\mathfrak{Q}$ 
all the clusters with shell index $k$.  As a consequence, $k$ is strictly
decreasing, and when the algorithm arrives to step~\ref{ws.endloop1}, either
$\mathfrak{Q}$ is empty or $C$ has a cluster verifying the hypothesis of
Corollary~\ref{conn3}.
\item The second {\bf while} loop will also finish because steps~\ref{kreset2} 
and~\ref{Qdeletion2} assure that $k$ is strictly decreasing.
\item The nested {\bf while} loops will finish because the 
families  $\mathfrak{Q}'$ and $\mathfrak{Q}$ are finite. 

\item 
Assume that the actual $C$ is core connected
in wide sense when we arrive to step~\ref{ws.qselec}.
By construction, $C\cap C_j=C$ when $j\leq k$,
and $C\cap C_j$ are previous instances of $C$ when $j>k$
(in fact, this instances are obtained each time
that we arrive to step~\ref{ws.endloop2}).
The new $C$, let us call it $C'$ for a moment, 
has the same intersections with $C_j$
for $j>k$, 
and when $j\leq k$ the intersection is $C'=C\cup Q$,
that is $k$ connected in wide sense by
Corollary~\ref{conn3},
as the conditions for the selection of $Q$ in step~\ref{ws.qselec}
match the hypothesis of Corollary~\ref{conn3}. 
Thus, all the instances of $C$
during the algorithm are core connected in
wide sense, 
whence the final $C$ is core connected in wide sense.
\end{itemize}
\end{proof}

\begin{figure}[!ht]
  \subfigure[Strict sense connectivity]{
    \includegraphics[width=0.55\textwidth]{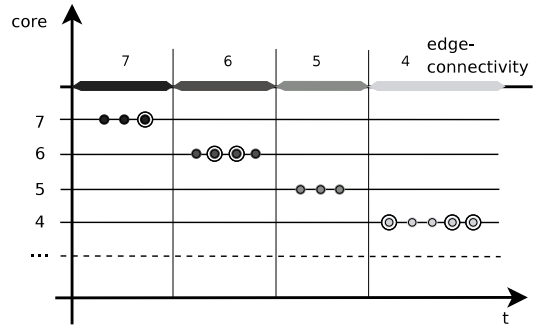}
  }
  \subfigure[Wide sense connectivity]{
    \includegraphics[width=0.55\textwidth]{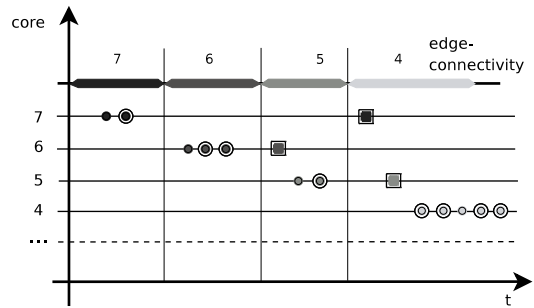}
  }
  \caption{\label{alg.il} 
    These two schemata show the progress in time ($t$)
of both algorithms.
Dots represent clusters, circled dots clusters in $C$, squared dots
clusters in $D$. Both algorithms start looking for a cluster of
diameter $2$, descending from the top shell. If they find someone,
this is the initial $C$. 
Then, the first algorithm looks, descending shell by shell, for clusters
to add to $C$, preserving core-connectivity in strict sense. 
The second algorithm does the same, but each time that it ends 
with the inspection of a shell, 
it looks between the omitted clusters from previous shells 
(saved in $\mathfrak Q'$), for clusters to add to $D$. 
These clusters provide more
possibilities of connection (in wide sense) to the next shell, thus producing
a bigger set $C$.   
    }
\end{figure}

Notice that in both algorithms we have limited $C$
to include nodes from the shells with $k \geq 2$.
The conditions from Corollary~\ref{conn3} are not
natural at level 1.
It is obvious that if $C$ is the strict sense
core-connected set constructed by Algorithm~\ref{ss_algorithm},
then after adjoining to $C$ all the clusters from
shell $1$ that are connected to it, the result
is also strict sense core-connected.
Analogously, it is clear that if $C$ and $D$ are the 
sets constructed by Algorithm~\ref{ws_algorithm},
if we adjoin to $C$ all the clusters from
shell $1$ connected to $C$ or $D$, then the new $C$
is also core-connected in wide sense.
We assume in the following Section that $C$ has been 
extended according to these remarks.

\section{Applications}\label{sec:appl}

In this section we will test our algorithms in some Internet maps obtained from different sources. Each of them has its own biases and explores the Internet at a particular level:
\begin{enumerate}
\item The Route Views Project \cite{oregon}, for instance, uses a short amount of BGP routers to peer with routers in other ASes and thus get routing tables. As a bias, this method does not detect hidden routes (not all inter-AS routes are public due to policies and agreements).

\item The CAIDA Association \cite{IR_CAIDA} developed {\em skitter} probes (which now evolved into the {\em Ark} infrastructure). Based on traceroute, this measurement nodes send ICMP requests to routers in order to discover paths. {\em Ark} has 38 monitors at November 2009. As an advantage, these tool finds the real routes followed by packets, advertised or not.
\item The DIMES Project \cite{DIMES} is a distributed system composed of around one thousand voluntary nodes --anyone may subscribe even with a low CPU power or bandwidth--. It explores the Internet with tools such as traceroute and ping to discover the topology.
\item Mercator \cite{govindan00heuristics} uses hop-limited ICMP probes to discover the Internet map with an {\em informed random} heuristic. The algorithm chooses some nodes in the path with the IP source-routing option, which is no longer available.
\end{enumerate}

At the Autonomous Systems interconnection level (called AS) data has been obtained from the Oregon Route Views Project, the CAIDA Association and the DIMES Project. Router level maps (IR) come from CAIDA, DIMES and Mercator.

To test our algorithms~\ref{ss_algorithm}
and~\ref{ws_algorithm} we implemented them into the LaNet-vi open source
software \cite{lanet-vi2}.  Based on the $k$-core decomposition, this tool
computes and visualizes the strict and wide core-connected components
giving also the lower bound of connectivity for every pair of nodes through a
logging file. This bound is obtained applying Corollary~\ref{conn3}
for nodes in the core-connected subgraph, and using the value of $\Phi$ for the others.
In the LaNet-vi color visualization (B\&W differences will be mentioned
in brackets), the nodes color (nodes shade) determines their core, while
the border color (border shade) suggests a lower bound for connectivity with other nodes. In fact, the absence of border points out that the node belongs to the core-connected subgraph and so its connectivity with nodes in inner shells from the core-connected subgraph is at least the node shell index. A colored border (shaded border), instead, implies that the node is not
in $C$ but it belongs to the $D$ set in
algorithm~\ref{ws_algorithm}, meaning that we assure certain level of
edge-connectivity with internal clusters in $C$ though this bound is lesser
than the node shell index. Finally for white nodes (squared nodes) we can
assure no edge-connectivity with others, and maybe the node is poorly connected.

\begin{figure}[!ht]
\begin{center}
\includegraphics[width=0.9\textwidth]{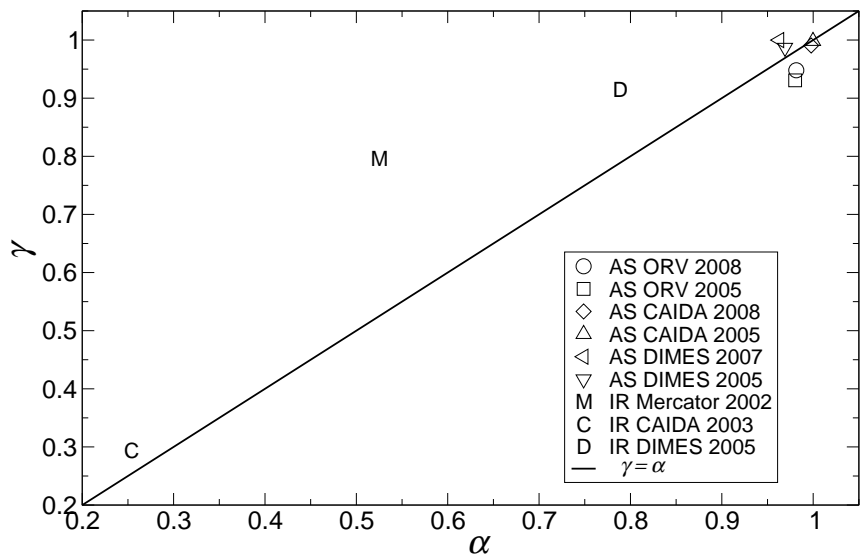}
\caption{Comparison between $\gamma$ and
$\alpha$. The line $\alpha=\gamma$ splits the space in two areas, the upper
one for graphs with $C$ concentrated on high populated shells ($\gamma>\alpha$), and the bottom for $C$ mainly in low populated
shells ($\gamma<\alpha$).
}\label{P_alpha}
\end{center}
\end{figure}
We performed different tests to show the effectiveness of our tool.
Defining the fraction of nodes in $C$ (the core-connected subgraph) for each shell $k$ as:
\begin{equation*}
\rho_k = \frac{|S_k \cap C|}{|S_k|} \enspace, 
\end{equation*}
we studied the following quantities:
\begin{eqnarray*}
 \alpha & = & \frac{1}{k_{\max}}\sum_{k=1}^{k_{\max}} \rho_k  \\
 \beta & = & \frac{2}{k_{\max} (k_{\max}+1)} \sum_{k=1}^{k_{\max}} \rho_k\; k \\
 \gamma & = & \sum_{k=1}^{k_{\max}} \rho_k \; \frac{|S_k|}{|C|} \enspace.\\
\end{eqnarray*}
In fact, $\alpha$, $\beta$ and $\gamma$ may be interpreted as 
probabilities for a node to be in $C$ under different models:
$\alpha$ is the average of all $\rho_k$; $\beta$ is a weighted average
of $\rho_k$, where each shell is weighted according to its number; and
$\gamma$ is a weighted average of the fractions $\rho_k$ with each
$k$-shell weighted by its size.

Figure~\ref{P_alpha} shows $\gamma$
as a function of $\alpha$. Networks with $\gamma>\alpha$ have most part of $C$ in high populated shells. We also present $\beta$ as a function of $\alpha$ in figure~\ref{beta_alpha}, where
$\beta>\alpha$ means that a lower part of $C$ is found in lower
shell indexes. 

On the one hand, AS maps are close to point $(1,1)$ in both figures,
and this means that most of their nodes belong to $C$. AS DIMES maps have $\beta<\alpha$
because some of their higher shells are empty.
But on the other hand, IR maps are worse core-connected than AS because the ratio of nodes in $C$ is low. Each IR has a different behavior; the most preferable
is IR DIMES ($90\%$ of nodes in $C$ and in higher shells).
The reason why the other IR maps give a small $C$ is the presence of big clusters in which the $\rho\leq 2$ condition is not satisfied. This is probably a bias in the exploration: IR DIMES maps
may be more accurate because they use a higher number of sources (thousands, according to \cite{DAHBVV2006}). We think that a detailed
picture of the Internet at the IR level will verify our hypotheses.
\begin{figure}[!ht]
\begin{center}
\includegraphics[width=0.9\textwidth]{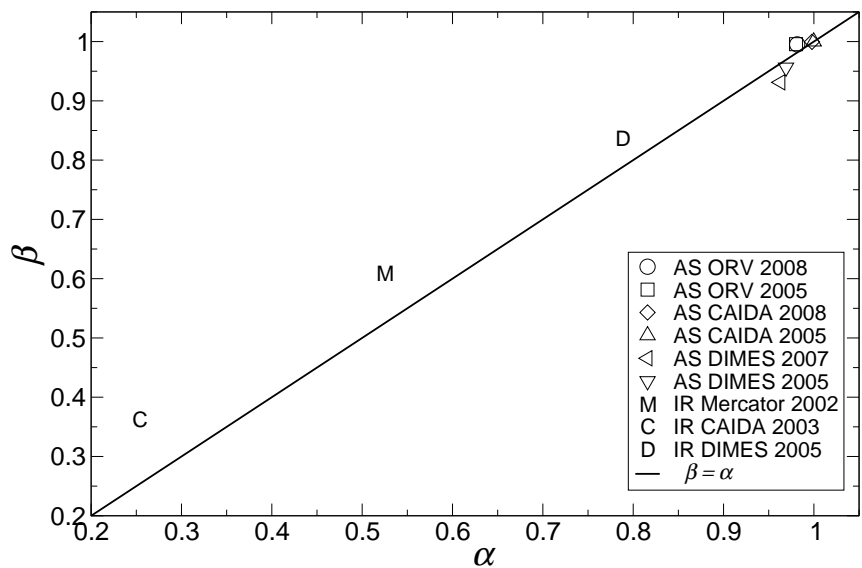}
\caption{$\beta$ vs. $\alpha$ for different maps.
}\label{beta_alpha}
\end{center}
\end{figure}

\begin{figure}[!ht]
\begin{center}
\includegraphics[width=0.9\textwidth]{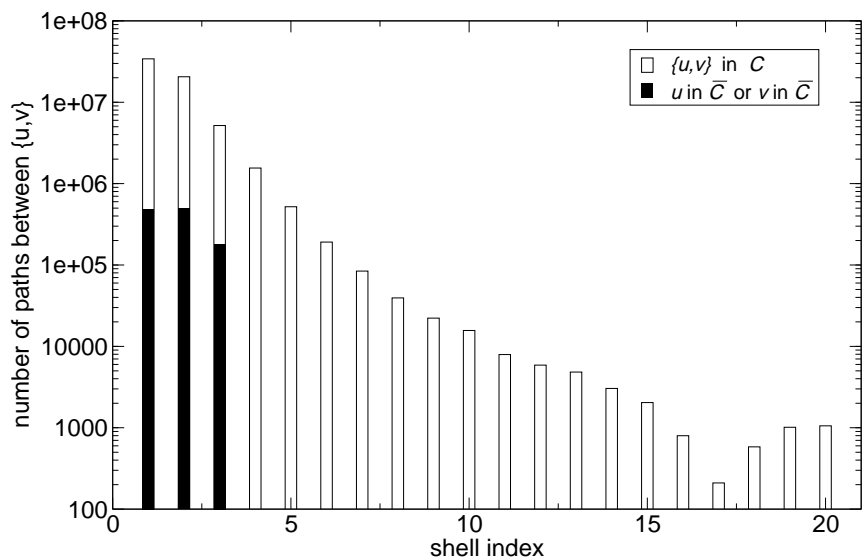}
\caption{Histogram showing the number of paths from $u$ to $v$ in the AS CAIDA 2008 map. white is for $\{u,v\}\in C$ and black is for $u\in \bar{C}$ or $v\in
\bar{C}$.}\label{histogram}
\end{center}
\end{figure}
A different analysis was done at shell level. At first we counted
pairs of nodes $\{u,v\}$ belonging to the wide-sense
core-connected graph, i.e., having both ends in $C$. Figure~\ref{histogram}
presents this information as a function of $min(sh(u),sh(v))$ for the AS
CAIDA 2008 map. It follows that only the first three shells have nodes
out of $C$, and that they are few (about $1\%$ per shell). This behavior is
similar in all maps but IR CAIDA 2003, where some nodes in low, medium
and high shells do not belong to the core-connected subgraph.

\begin{figure}[!ht]
\begin{center}
\includegraphics[width=0.9\textwidth]{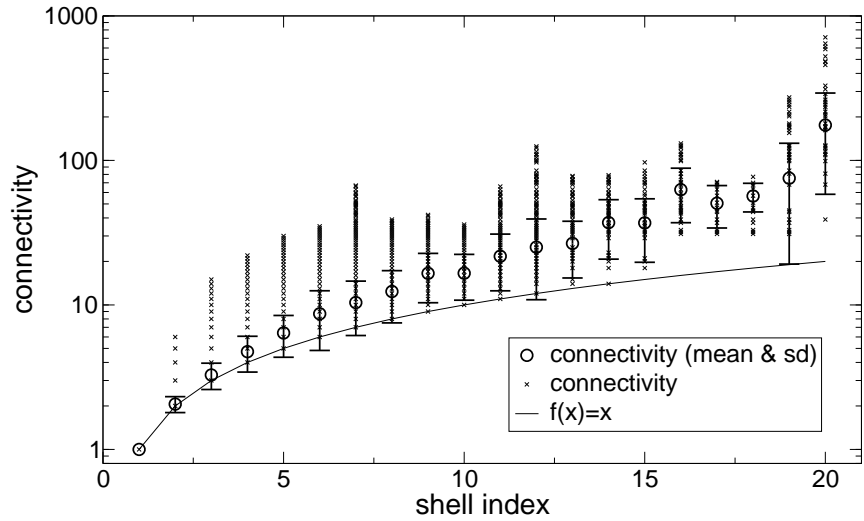}
\caption{Connectivity for $\{u,v\}$ pairs vs. their minimum
shell index. Circles stand for mean values per shell index and error bars show the standard
deviation; crosses show connectivity;
the bounding line {\em connectivity=shell index} is also displayed.}\label{gomory-hu-cmp}
\end{center}
\end{figure}
\begin{figure}[!ht]
\begin{center}
\includegraphics[width=0.9\textwidth]{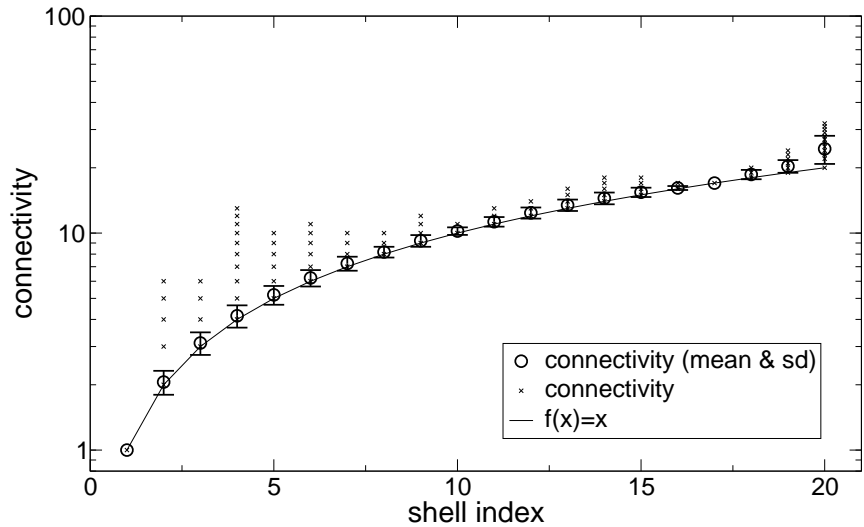}
\caption{Connectivity through the cores for $\{u,v\}$ pairs vs. their minimum shell index. Circles stand for mean values and error bars
show the standard deviation; crosses show connectivity;
the bounding line {\em connectivity=shell index} is also
displayed.}\label{gomory-hu-cmp-cores}
\end{center}
\end{figure}
Secondly, we compared the minimum shell index of a pair of nodes
$\{u,v\}$ in a core-connected subgraph with the connectivity, this last
obtained with the Gomory-Hu algorithm (see chapter 4 in \cite{ffbook}).
In figure~\ref{gomory-hu-cmp} we display the connectivity of each pair
$\{u,v\}$ as a function of the minimum shell index $min(sh(u),sh(v))$,
for the AS CAIDA 2008 map. We see that connectivity is
higher than the bound obtained by the core-connected graph, still the
average values are relatively close for low and medium shells (less than $100\%$ up to shell 13). As our bound is related to connectivity through cores (i.e., taking only the maximum $k$-core
containing $\{u,v\}$ to find paths) we show also the connectivity
through the $k$-core, where $k=min(sh(u),sh(v))$ in
Figure~\ref{gomory-hu-cmp-cores}. Clearly, connectivity is larger than connectivity through
the $k$-core, but this last is closer
to the minimum shell index. The other AS maps have analogous behavior.

As a remark, we couldn't compute connectivity for IR maps because it is
expensive on RAM memory for their size, which is greather than 100,000 nodes.

\begin{figure}[!ht]
\begin{center}
\includegraphics[width=\textwidth]{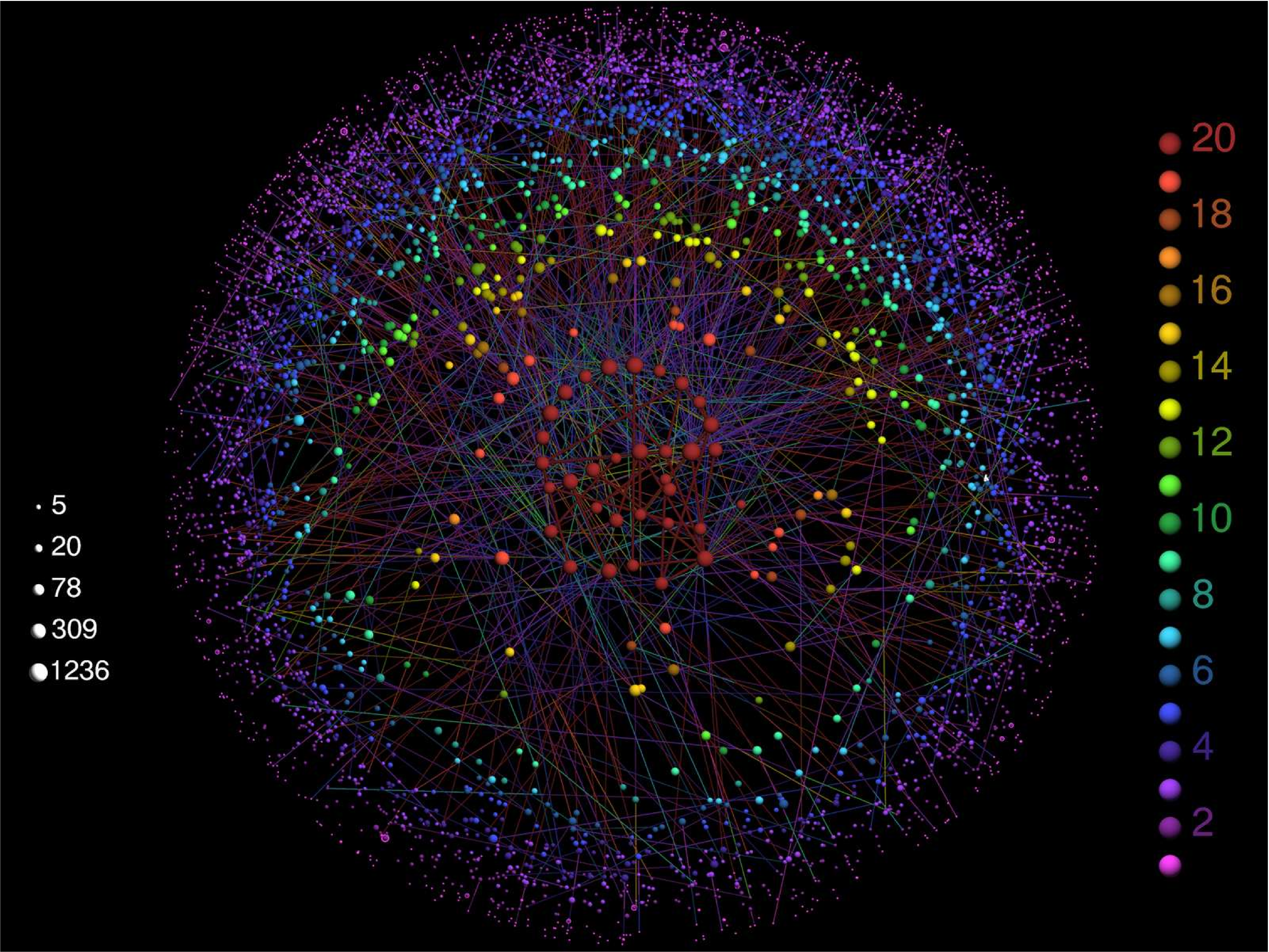}
\includegraphics[width=\textwidth]{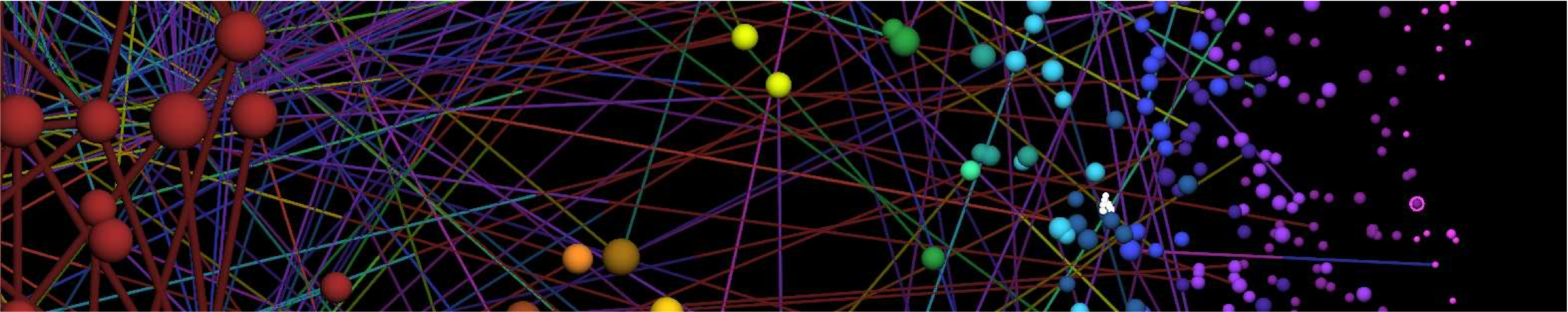}
\caption{Color visualization of AS CAIDA 2008 map displaying the wide-sense
core-connected subgraph (algorithm~\ref{ws_algorithm}). Bottom: a detail
showing some nodes out of the core-connected subgraph, and some in
white (out of $C$ and $D$); we can see one node belonging to $D$ on the
right.}\label{visAC2008}
\end{center}
\end{figure}
\begin{figure}[!ht]
\begin{center}
\includegraphics[width=\textwidth]{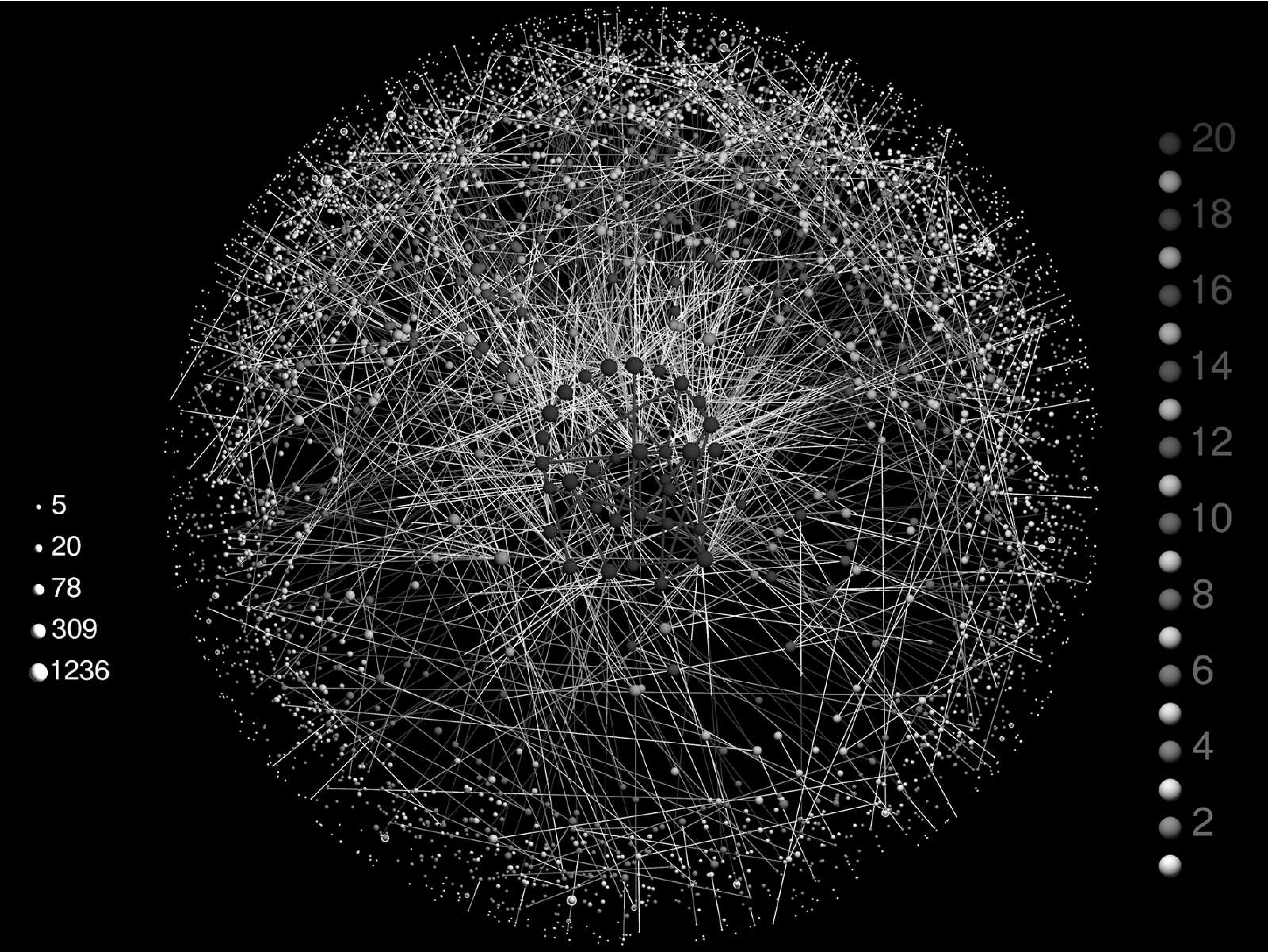}
\includegraphics[width=\textwidth]{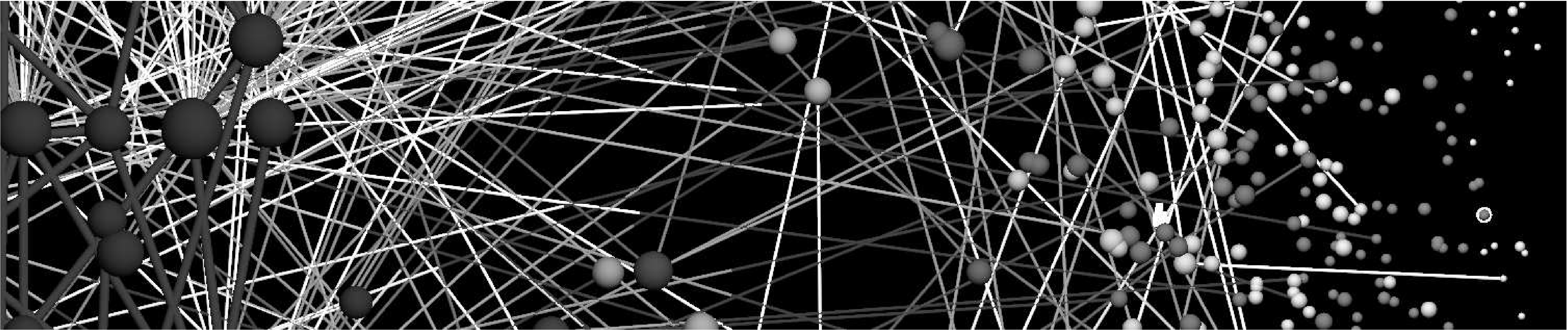}
\caption{Greyscale visualization of AS CAIDA 2008 map displaying the
wide-sense core-connected subgraph (algorithm~\ref{ws_algorithm}).
Bottom: a detail
showing some nodes out of the core-connected subgraph, and some in
white (out of $C$ and $D$); we can see one node belonging to $D$ on the
right.}\label{visAC2008BN}
\end{center}
\end{figure}
To conclude our analysis, we present in figure~\ref{visAC2008} a
color visualization for AS CAIDA 2008 map (greyscale is shown in
figure~\ref{visAC2008BN}), testing the algorithm~\ref{ws_algorithm}. 
%

\section{Conclusions}

In this work we have defined core-connected graphs, for which we obtain a lower bound of connectivity between nodes: the minimum shell index of them in the graph $k$-core decomposition. 
To formalize this relation we provided a theorem (Corollary 2) that gives
sufficient conditions to assure the forementioned lower bound. This
theorem is an extension of Plesnik's theorem
(see \cite{plesnik:cgoagd}):

Plesnik's theorem asserts that if a simple graph has diameter $2$, then
the connectivity is at least $\delta(G)$. It is not difficult to show
that the hypotheses of $G$ being simple can be easily relaxed, thus
obtaining that the connectivity is at least $\min_{v\in V} |N(v)|$
(which is $\delta(G)$ in the simple case).  The combinatorial nature of
connectivity, where multiple bifurcations give rise to a great
multiplicity of paths between two vertices, in the absence of bottle
necks, justifies the presumption that the diameter $2$ condition is
rather artificial.  In fact, \cite{TRmedusa06,medusa07} have noticed that the node to node
connectivity $k'(u,v)$ is at least $\min(sh(u),sh(v))$ for all but very
exceptional pairs $u,v$ in many real life net graphs (indeed, the
connectivity in this paper is less than $k'$, because the authors count
only disjoint paths, where disjoint means that they do not share neither
edges nor vertices).  Our results herein show some semilocal conditions
under which this bound for the connectivity holds, where semilocal means
here that the conditions involve, for each $k$, the relations between a
$k$-core and his next $(k-1)$-shell.
Rather that an alternative procedure to find the connectivity, which can
be rather efficiently found with the Gomory-Hu algorithm (see chapter 4 of
\cite{ffbook}), we hope that our results give some new
insight on the local-global relations for connectivity, useful in real
life net graphs.

We also developed two algorithms to get core-connected subgraphs of a given graph $G$: one for strict-sense connectivity (whose complexity is $O(e)$) and 
one for wide-sense connectivity ($O(e \times \sqrt{e})$).
In the strict-sense algorithm each cluster is considered only once to
determine if it fulfills conditions about its diameter and $\Phi$. Being
$e(Q)$ the amount of edges of cluster $Q$ counting the internal ones and
the ones that connect it with inner cores (clusters don't share edges
with other clusters in the same core), then in $O(e(Q))$ it can be
determined if the cluster's diameter is less or equal to $2$. The $\Phi$ condition can also be evaluated in $O(e(Q))$ as it implies a BFS in the cluster. Consequently, covering all clusters, the algorithm runs in $O(e)$.
In the wide-sense algorithm each cluster may be considered up to $k_{max}$ times (once per each shell). But $k_{max}$ is bounded above by $\sqrt{e}$ because to obtain a $k$-shell, $k+1$ nodes are needed at least with $k$ connections each, so the graph must have $k \times (k-1)$ edges. Then we get a complexity of $O(e \times \sqrt{e})$.

Finally we included these algorithms in the open-source software LaNet-vi to visualize core-connected subgraphs and list nodes in them, showing that it works for the Internet maps.

We are working on a possible relaxation on the conditions involved with $\Phi$ and the treatment of clusters with $\rho> 2$. We will also look for explanations on the high connectivity found in the higher shells.

\bibliographystyle{apalike}
\bibliography{core-connected}

\end{document}